\documentclass[10pt,final,conference,letterpaper,oneside,twocolumn]{ieeeconf}
\IEEEoverridecommandlockouts
\usepackage{fixltx2e}

\usepackage{amssymb}
\usepackage[cmex10]{amsmath}
\usepackage{amstext}
\usepackage{array}
\usepackage{amsthm}
\usepackage{quoting}
\usepackage{tikz}
\usepackage{pgfplots}
\usepackage{import}
\usepackage{mathtools}
\usepgfplotslibrary{units}
\usetikzlibrary{spy,backgrounds}
\usepackage{pgfplotstable}

\usepackage[hidelinks,bookmarks=false]{hyperref}
\usepackage[figure]{hypcap}
\usepackage{indentfirst}
\usepackage{float}
\usepackage{cite}
\usepackage{microtype}
\usepackage[T1]{fontenc}
\usepackage[font=normalsize,labelfont=normalsize]{caption}
\theoremstyle{definition}
\newtheorem{ass}{Assumption}[]
\newtheorem{defin}{Definition}[]
\newtheorem{theo}{Theorem}[]
\newtheorem{rem}{Remark}[]
\newlength\figureheight
\newlength\figurewidth

\makeatletter
\newcommand{\Eqref}[1]{\textup{\tagform@{\ref*{#1}}}}
\makeatother

\newcommand\copyrighttext{%
\Huge \textcopyright 2017 IEEE. Personal use of this material is permitted. Permission from IEEE must be obtained for all other uses, in any current or future media, including reprinting/republishing this material for advertising or promotional purposes, creating new collective works, for resale or redistribution to servers or lists, or reuse of any copyrighted component of this work in other works.}
\newcommand\copyrightnotice{%
\begin{tikzpicture}[remember picture,overlay]
\node[anchor=center] at (current page.center) {\fbox{\parbox{\dimexpr\textwidth-\fboxsep-\fboxrule\relax}{\copyrighttext}}};
\end{tikzpicture}%
}

\pgfplotsset{
compat=newest,
every axis/.append style={
font=\tiny
},
tick style={thin,black},
minor tick style={very thin,gray},
every axis legend/.append style={
draw=black,
fill=white,
shape=rectangle,
cells={anchor=center},
inner xsep=3pt,
inner ysep=2pt,
font=\tiny
}}

\begin{document}

\copyrightnotice

\title{\LARGE \bf
Persistently Exciting Tube MPC
}
\author{Bernardo Hernandez$^{1,a}$ and Paul Trodden$^{2}$%
\thanks{$^{1}$
Bernardo Hernandez is a PhD student at the Dept. of Automatic Control and Systems Engineering, The University of Sheffield, Sheffield, S1 4DW, UK {\tt\small \href{mailto:bahernandezvicente1@sheffield.ac.uk}{bahernandezvicente1@sheffield.ac.uk}}}%
\thanks{$^{2}$Paul Trodden is with the Dept. of Automatic Control and Systems Engineering, The University of Sheffield, Sheffield, S1 4DW, UK {\tt\small \href{mailto:p.trodden@sheffield.ac.uk}{p.trodden@sheffield.ac.uk}}}%
\thanks{$^{a}$Bernardo Hernandez acknowledges doctoral scholarship from CONICYT - PFCHA/Concurso para Beca de Doctorado en el Extranjero - 72150125.}%
}
\maketitle

\begin{abstract}
\boldmath This paper presents a new approach to deal with the dual problem of system identification and regulation. The main feature consists in breaking the control input to the system into a regulator part and a persistently exciting part. The former is used to regulate the plant using a robust MPC formulation, in which the latter is treated as a bounded additive disturbance. The identification process is executed by a simple recursive least squares algorithm. In order to guarantee sufficient excitation for the identification, an additional non-convex constraint is enforced over the persistently exciting part.
\end{abstract}

\section{Introduction} \label{sec:Intro}
The performance and optimality of a model predictive controller are limited by the accuracy of the model used to make the predictions \cite{Adetola2011,Aswani2013}. If the model is a poor representation of reality, the resulting control actions will not be optimal for the latter, causing unexpected behaviour. Moreover, desirable properties in model predictive control (MPC) formulations, such as stability and feasibility \cite{Mayne2013,Boccia2014}, often require the computation of invariant sets, which are model dependent \cite{Kolmanovsky}. Adaptive MPC, as a way to cope with model uncertainty, has been receiving increasing amounts of attention from researchers in the last ten years, but it remains a largely open problem \cite{Mayne2014}. Adaptive MPC attempts to couple closed-loop system identification with regulation through an MPC controller; an inherent difficulty of such design is that the two objectives are incompatible. This is referred to as the \emph{dual control problem}: while the controller tries to steady the system, the identifier needs to excite it \cite{Goodwin1984}.

A key challenge for adaptive MPC is how to maintain the stability and feasibility guarantees, particularly when hard constraints are considered. In \cite{Mamboundou2011} the model of an unconstrained plant is updated through a modified recursive least squares (RLS) algorithm, while a fuzzy supervisor attempts to update the controller parameters based on performance criteria. These include a numeric evaluation of the system stability, but no guarantee is given. \cite{Balderud2011} uses a single value decomposition based algorithm to estimate, on-line, a state space model of the controlled system; stability is shown solely through numerical simulations. A set membership identification scheme is used in \cite{Tanaskovic2014}, coupled with additional output constraints in the optimization problem. These extra constraints ensure boundedness of the system response and hence constraint satisfaction, however open-loop stability is assumed.
Other authors have addressed the issue of stability by making suitable assumptions, such as in \cite{Fukushima2007,Adetola2011,Wang2014a}, where the response error produced by the uncertainty of the model is treated as a bounded disturbance. This allows for robust MPC implementations to be used under some supplementary assumptions, such as a known bound on the initial estimates error \cite{Fukushima2007}, \cite{Adetola2011}. A novel algorithm is developed in \cite{Aswani2013}, where two models of the plant are maintained by the controller; a nominal model is used to provide feasibility and stability guarantees, while the second, adaptive, model is used to improve performance.
However, none of these approaches consider the dual control problem in an explicit way. In \cite{Fukushima2007} the identifier chooses suitable past data for the computation of estimate bounds, but nothing ensures the existence of this data. In \cite{Adetola2011} a min-max robust approach accounts for future estimation but does not directly promote proper excitation.

In the MPC context, this matter has been addressed in different ways. In \cite{Genceli1996,Vuthandam1997a, Shouche1998,Marafioti2013,Larsson2015,Heirung2015}, an additional constraint over the input is explicitly added to guarantee enough information on the output.
However, feasibility of the optimization may be lost due to this constraint, which is in most cases non-convex. To overcome this, a relaxation variable is employed in \cite{Genceli1996,Shouche1998}, yet this may result in insufficient excitation in favour of feasibility. The receding horizon fashion of MPC (i.e., that only the first part of the optimized input sequence is applied) is exploited in \cite{Marafioti2013}, where the additional constraint is applied exclusively over the first element of the predicted input sequence; this concept is also studied in \cite{Heirung2015}, where it is shown that extending this type of constraint beyond one time step into the prediction horizon has a negligible impact on the identification results.
Alternatively, a two-step optimization is performed in \cite{Rathousky2013,Zacekova2013}; the first step solves a standard MPC problem and the second step adds an exciting behaviour to the optimized sequence, while limiting the cost increase (reduced optimality). The concept of zone-tracking MPC is used in \cite{Ferramosca2013a} to drive the state of the plant to an invariant set, inside which a persistently exciting input sequence can be safely applied. In recent implementations \cite{Weiss2014,Heirung2015}, the MPC cost function is augmented with a term depending on the variance of the estimated parameters error, in an attempt to force the optimizer to choose an exciting input sequence, that automatically decays as the system is correctly identified.

In this paper, the dual problem of regulation and system identification is addressed within the frame of robust MPC; the main feature of the present algorithm is the division of the input signal. The first part, called the persistently exciting (PE) part, aims to generate enough information for the identification process, while the second (regulator) part is designed following the main objective of regulating the plant to the desired steady state. From the control perspective, the PE part of the input may be treated as a bounded disturbance, hence a standard tube MPC formulation \cite{Mayne2005} is a suitable selection for the regulation task. This allows to maintain the classic form of the optimization problem, unlike \cite{Rathousky2013,Zacekova2013,Weiss2014,Heirung2015}, and also helps to provide a guarantee of asymptotic stability. 
A standard RLS algorithm \cite{Ljung1999} with forgetting factor is used for the identification process. To achieve convergence of the estimates, an additional constraint based on the persistence of excitation theory \cite{Goodwin1984} is included in the optimization. This ensures not only an accurate estimation but also that the PE input is automatically defined by the optimizer (contrary to the approach in \cite{Ferramosca2013a}, where it must be computed off-line). The PE constraint proposed in \cite{Marafioti2013} is tightened so as to guarantee recursive feasibility while allowing for increased optimality.
The disturbances considered at this stage arise from model mismatch and the excitation part of the input, thus the approach is inherently deterministic.

The paper is organized as follows: Section \ref*{sec:PS} defines the problem and required preliminaries. Section \ref*{sec:PETMPC} describes the proposed persistently exciting tube MPC (PE-Tube MPC). Stability and feasibility are established in Section \ref*{sec:SF}. Section \ref*{sec:SR.1} contains numerical simulation results.

\emph{Notation}: The operators $\oplus$ and $\ominus$ denote the Minkowski sum and the Pontryagin difference (respectively), defined as in \cite{Rawlings2014}. The set $\mathbb{N}_{0}$ is the set of all the positive integers including 0. The zero vector and the identity matrix in $\mathbb{R}^n$ are represented respectively by $\boldsymbol{0}_n$ and $I_n$.

\section{Problem statement and preliminaries} \label{sec:PS}
The problem is to regulate a linear time invariant system, subject to input and state constraints, for which only a nominal discrete time state space model is available. Define,
\begin{subequations} \label{eq:PS.0}
	\begin{alignat}{2}
	&\left(\bar{A},\bar{B}\right)&&\quad\text{Nominal model.} \label{eq:PS.0.1}\\
	&\left(A,B\right)&&\quad\text{Real system.} \label{eq:PS.0.2}\\
	&\left(\tilde{A}(i),\tilde{B}(i)\right)&&\quad\text{Estimated model at time $i$.}\label{eq:PS.0.3}
	\end{alignat}
\end{subequations}
The nominal state space model takes the following form,
\begin{subequations} \label{eq:PS.1}
\begin{alignat}{1}
x(i+1)&=\bar{A}x(i)+\bar{B}u(i) \label{eq:PS.1.1}\\
x(i) \in \mathbb{X}\subset\mathbb{R}^n,&\quad u(i) \in \mathbb{U}\subset\mathbb{R}^m,\qquad\forall i \in \mathbb{N}_{0} \label{eq:PS.1.3}
\end{alignat}
\end{subequations}
in which $x(i)$ and $u(i)$ are the state and input vectors at time $i$. The following common assumptions are supposed to hold.

\begin{ass}[Stabilizability] \label{ass:PS.1}
The pairs $\left(A,B\right)$, $\left(\bar{A},\bar{B}\right)$ and $\left(\tilde{A}(i),\tilde{B}(i)\right)$ are stabilizable $\forall i\in\mathbb{N}_{0}$.
\end{ass}

\begin{ass}[Properties of constraint sets] \label{ass:PS.3}
The set $\mathbb{X}$ is closed and the set $\mathbb{U}$ is compact. Both contain the origin.
\end{ass}

\subsection{Standard MPC formulation} \label{sec:PS.MPC}
The standard MPC optimization problem for the system in \Eqref{eq:PS.1} with prediction horizon $N$ at time $i$ is $\mathbb{P}_N\big(x=x(i)\big)$:
\begin{equation} \label{eq:PS.2}
\min_{u}\sum_{k=0}^{N-1}\underbrace{\left(x^{\top}(k)Qx(k)+u^{\top}(k)Ru(k)\right)}_{\ell(x,u)} + V_{f}\big(x(N)\big)
\end{equation}
subject to:
\begin{subequations} \label{eq:PS.3}
\begin{alignat}{2}
x(0)&=x&&  \label{eq:PS.3.0}\\
x(k+1)&=\bar{A}x(k)+\bar{B}u(k)&&  \label{eq:PS.3.1}\\
x(k) \in \mathbb{X},&\quad u(k) \in \mathbb{U},&&\quad k=0,1,\ldots,N-1  \label{eq:PS.3.3}\\
x(N) &\in \mathbb{X}_{f}\subseteq \mathbb{X}&&  \label{eq:PS.3.4}
\end{alignat}
\end{subequations}
where $(Q,R)$ are the state and input weight matrices. Once the optimization is solved, the first part of the optimal input sequence is applied to the plant, a new state measurement is taken, and the process is repeated.

It is well known that, under Assumptions \ref*{ass:PS.1} and \ref*{ass:PS.3}, an appropriate selection of the weight matrices, terminal cost $V_{f}(\cdot)$ and terminal constraint $\mathbb{X}_{f}$ provides closed-loop asymptotic stability of the origin \cite{Rawlings2014} for the nominal system. In particular, we use the following standard assumption:

\begin{ass}[Stability assumption] \label{ass:PS.2}
$V_{f}:\mathbb{X}_{f}\rightarrow\mathbb{R}_{\geq0}$ is continuous and $V_{f}(\boldsymbol{0}_n)=0$. $Q$ is positive semidefinite and $R$ is positive definite. The set $\mathbb{X}_f$ is a closed control invariant set for \eqref{eq:PS.3.1}, contains the origin in its interior, and,
\begin{equation*} \label{ass:PS.2.1}
\exists u\in\mathbb{U}\;\mbox{s.t.}\;V_{f}\left(\bar{A}x+\bar{B}u\right)+\ell(x,u)\leq V_{f}(x),\;\;\forall x\in\mathbb{X}_{f}
\end{equation*}
\end{ass}

If the initial state is inside the feasibility region defined by constraints \eqref{eq:PS.3.3} and \eqref{eq:PS.3.4} (called from now on $\mathcal{X}_{N}$), recursive feasibility of the optimization can be guaranteed \cite{Rawlings2014}.

\subsection{Persistence of excitation} \label{sec:PS.PE}
For many reasons, the nominal model $\left(\bar{A},\bar{B}\right)$ may not be an accurate representation of the real system $\left(A,B\right)$. This could have a detrimental effect on the performance and stability of the MPC controlled system; therefore, to reduced model uncertainty, some form of closed-loop system identification should be implemented. Note that a state space model may be regarded as a system of $ARX(1,1)$ models, hence a predictor can be built for each component of the state vector,
\begin{subequations} \label{eq:PS.5}
\begin{alignat}{2}
\hat{x}_{j}(i)&=\phi^{\top}(i-1)\tilde{\theta}_{j}(i),\quad && j=1,2,\ldots,n \label{eq:PS.5.1}\\
\phi^{\top}(i)&=\left[x^{\top}(i) \:\: u^{\top}(i)\right] \label{eq:PS.5.2}\\
\tilde{\theta}_{j}(i)&=\left[\tilde{A}_{j}(i) \:\: \tilde{B}_{j}(i)\right]^{\top},\quad && j=1,2,\ldots,n \label{eq:PS.5.3}
\end{alignat}
\end{subequations}
In \Eqref{eq:PS.5}, $\hat{x}_j(i)$ represents the prediction of the state component $j$, at time $i$, $\phi(i)$ is the regressor vector and $\left(\tilde{A}_{j}(i),\tilde{B}_{j}(i)\right)$ are the $j^{\text{th}}$ rows of the currently estimated matrices \eqref{eq:PS.0.3}. A standard RLS algorithm with constant forgetting factor $\lambda$ \cite{Ljung1999} is employed to identify a new model every time step. The recursion at time $i$ is computed as follows,
\begin{subequations} \label{eq:SF.5}
\begin{alignat}{2}
\tilde{\theta}_{j}(i)=&\tilde{\theta}_{j}(i-1)+R^{-1}_{\text{ID}}(i)\phi(i)\big[x_{j}&&(i)\notag\\
&-\phi^{\top}(i-1)\tilde{\theta}_{j}(i-1)\big],&&\quad j=1,2,\ldots,n\label{eq:SF.5.1}\\
R_{\text{ID}}(i)=&\lambda R_{\text{ID}}(i-1)+\phi(i)\phi^{\top}(i),&&\quad j=1,2,\ldots,n\label{eq:SF.5.2}
\end{alignat}
\end{subequations}
Convergence of the estimated parameters $\tilde{\theta}_{j}$ is guaranteed if the regressor is a strongly persistently exciting sequence \cite{Goodwin1984}.

\begin{defin}[Strongly persistently exciting sequence] \label{defin:PS.1}
The sequence $\{\phi(i)\}=\phi(0),\phi(1),\cdots,\phi(i)$, is said to be \emph{strongly persistently exciting} of order $N_p$ at time $i$, if there exists an integer $l_p$ and real numbers $\rho_{0},\rho_{1}>0$ such that,
\begin{equation*} 
\begin{aligned}
\rho_{1}I_{(n+m)N_p}>&\sum_{j=0}^{l_p-1}\left(\boldsymbol{\phi}_{i-j}\boldsymbol{\phi}^{\top}_{i-j}\right)>\rho_{0}I_{(n+m)N_p} \\
\boldsymbol{\phi}_{i-j}&=\left[\begin{array}{c}
\phi(i-j)\\
\phi(i-j-1)\\
\vdots\\
\phi(i-j-Np+1)
\end{array} \right] 
\end{aligned}
\end{equation*}
\end{defin}

The variable $N_p$ defines the length of a time window that is going to be observed and the variable $l_p$ defines the number of time instants into the past that this window will be observed. Definition \ref*{defin:PS.1} is identical to definition 3.4.A given in \cite{Goodwin1984} but after a time shift, which objective is to set the current time $i$ as the upper time limit (i.e. the window is placed at time $i$ and it moves backwards). In this way, coupling with the receding horizon fashion of MPC is achieved in a straightforward way.

Persistence of excitation of the regressor vector is not a suitable condition to use as a constraint in the MPC context, mainly because the state vector is not an explicit decision variable of the optimization \Eqref{eq:PS.2}--\Eqref{eq:PS.3}. Within the MPC framework, it is more convenient to focus on the input, and how the persistence of excitation propagates to the regressor. To do this, the concept of state reachability is employed.

\begin{defin}[State reachability] \label{defin:PS.2}
System \Eqref{eq:PS.1} is said to be \emph{state reachable} if, for any $x\in\mathbb{X}$, there exists an input sequence $\{u(j)\in\mathbb{U}\}_{j=0:s<\infty}\ $ such that at time $s$, $x(s)=x$.
\end{defin}

\begin{theo}[Persistence of excitation of reachable systems] \label{theo:PS.1}
The sequence $\{\phi(i)\}=\phi(0),\phi(1),\cdots,\phi(i)$, with $\phi(\cdot)$ defined as in \Eqref{eq:PS.5.2}, is said to be \emph{strongly persistently exciting} of order $N_p$ at time $i$ if, the system \Eqref{eq:PS.5.1} is state reachable and there exists an integer $l_p$ and real numbers $\rho_{0},\rho_{1}>0$ such that,
\begin{subequations} \label{eq:PS.7}
\begin{align}
\rho_{1}I_{mN_p}>&\sum_{j=0}^{lp-1}\left(\boldsymbol{u}_{i-j}\boldsymbol{u}^{\top}_{i-j}\right)>\rho_{0}I_{mN_p} \label{eq:PS.7.1}\\
\boldsymbol{u}_{i-j}&=\left[\begin{array}{c}
u(i-j)\\
u(i-j-1)\\
\vdots\\
u(i-j-N_p+1) \label{eq:PS.7.2}
\end{array} \right]
\end{align}
\end{subequations}
\end{theo}
\begin{proof}
This proof can be found in \cite{Green1985} (Theorem 2.1).
\end{proof}

In \cite{Marafioti2013} it is shown that the lower bound of \eqref{eq:PS.7.1} characterizes the outside of an ellipsoid, hence the PE constraint is non-convex. Also note that $\boldsymbol{u}_{i-j}=\left[\boldsymbol{0}_{m(N_p-1)}\right]$ (or any fixed value) violates \eqref{eq:PS.7.1}, therefore regulation to a steady state and persistence of excitation cannot be simultaneously attained.

\section{Tube MPC with persistence of excitation} \label{sec:PETMPC}
The main contribution of this paper is presented in this section, the underlying idea is to include a persistence of excitation constraint in a standard MPC formulation. This is done from a robust control perspective, where the excitation is regarded as a bounded disturbance. The proposed approach uses tube MPC, which is a robust control technique with guaranteed stability under bounded additive uncertainties, but complexity similar to conventional MPC \cite{Rawlings2014,Mayne2005}.

\subsection{Tube MPC for uncertain systems with partitioned input} \label{sec:PETMPC.TMPC}
Tube MPC solves the regulation problem for an undisturbed nominal model, while securing that the state of the uncertain system will always be in a robust positive invariant (RPI) set \cite{Kolmanovsky}, centred around the nominal system trajectory. This robust control technique is inherently capable of dealing with model uncertainties as long as these can be \emph{quantified}, i.e., treated as a bounded additive disturbance. This requires a certain insight on how different the real system \Eqref{eq:PS.0.1} and prediction model \Eqref{eq:PS.0.3} may be. To account for this, the following assumption is supposed to hold,

\begin{ass}[Size of parametric uncertainty] \label{ass:PETMPC.0}
A set $\mathbb{W}_S\coloneqq\big\{w_S=(A-\tilde{A}(i))x+(B-\tilde{B}(i))u\mid(x,u)\in\mathbb{X}\times\mathbb{U},\quad\forall i\in\mathbb{N}_{0}\big\}$ is known.
\end{ass}

Consider the model structure of \Eqref{eq:PS.1}. Henceforth, the input will be divided into a regulator part, $\hat{u}$, and a persistently exciting part, $w$. The nominal model \Eqref{eq:PS.1} is rewritten as,
\begin{subequations} \label{eq:PETMPC.1}
\begin{alignat}{2}
x(i+1)&=\bar{A}x(i)+\bar{B}\left(\hat{u}(i)+w(i)\right) \label{eq:PETMPC.1.1}\\
x(i) \in \mathbb{X},&\quad\hat{u}(i) \in \hat{\mathbb{U}},\quad w(i) \in \mathbb{W},\quad&&\forall i \in \mathbb{N}_{0} \label{eq:PETMPC.1.4}\\
\bar{B}w(i)&=\hat{w}(i)\in \hat{\mathbb{W}}=\bar{B}\mathbb{W},\quad&&\forall i \in \mathbb{N}_{0} \label{eq:PETMPC.1.5}
\end{alignat}
\end{subequations}

\begin{ass}[Properties of the divided input constraint sets] \label{ass:PETMPC.1}
The sets $\mathbb{W}$ and $\hat{\mathbb{U}}$ are compact and contain the origin. Also $\hat{\mathbb{U}}\subseteq\mathbb{U} \ominus \mathbb{W}$ and is non-empty.
\end{ass}
Assumption \ref*{ass:PETMPC.1} implies $\hat{\mathbb{U}}\oplus\mathbb{W}\subseteq\mathbb{U}$. Note that $\hat{\mathbb{W}}$ is a linear mapping of $\mathbb{W}$ therefore it maintains compactness \cite{Sutherland1995}.

The tube MPC approach treats $\hat{w}$ as a bounded additive disturbance. The undisturbed model takes the form,
\begin{subequations} \label{eq:PETMPC.2}
\begin{alignat}{1}
z(i+1)&=\bar{A}z(i)+\bar{B}v(i) \label{eq:PETMPC.2.1}\\
z(i)\in \mathbb{Z}&=\mathbb{X}\ominus \mathbb{S},\:\:\: v(i) \in \mathbb{V}=\hat{\mathbb{U}}\ominus K_t\mathbb{S},\quad\forall i \in \mathbb{N}_{0} \label{eq:PETMPC.2.3}
\end{alignat}
\end{subequations}
in which $K_t$ is any stabilizing gain for $(\bar{A},\bar{B})$, guaranteed to exist in view of Assumption \ref*{ass:PS.1}. The set $\mathbb{S}$ is an RPI set for the dynamics of the error between the trajectories of the nominal and uncertain models,
\begin{subequations} \label{eq:PETMPC.3}
\begin{alignat}{1}
e(i)&=x(i)-z(i),\quad e(0)=0 \label{eq:PETMPC.3.1}\\
e(i+1)&=(\bar{A}+\bar{B}K_t)e(i)+\hat{w}(i)+w_S(i) \label{eq:PETMPC.3.2}\\
\hat{w}(i) &\in \hat{\mathbb{W}},\quad w_S(i) \in \mathbb{W}_S,\quad\forall i \in \mathbb{N}_{0} \label{eq:PETMPC.3.5}
\end{alignat}
\end{subequations}
Note that the undisturbed system evolves independently after $i=0$. For a nominal solution to exist, the constraint space defined by constraints \Eqref{eq:PETMPC.2.3} must be non-empty,

\begin{ass}[Allowable disturbance size for constraint satisfaction] \label{ass:PETMPC.2}
The set $\mathbb{S}$ is such that $\mathbb{S}\subset\mathbb{X}$ and  $K_t\mathbb{S}\subset\hat{\mathbb{U}}$.
\end{ass}

Assumption \ref*{ass:PETMPC.2} is not uncommon in robust control implementations, it simply states the fact that it will not be possible to satisfy the constraints if the disturbances are \emph{too} large. After a standard MPC problem (section \ref*{sec:PS.MPC}) is solved for the nominal system \Eqref{eq:PETMPC.2.1} under tightened constraints \Eqref{eq:PETMPC.2.3}, the input to the uncertain system is computed from the following control policy,
\begin{equation} \label{eq:PETMPC.3.3}
\hat{u}(i)=v(i)+K_t\left(x(i)-z(i)\right)
\end{equation}

\subsection{Additional PE constraint} \label{sec:PETMPC.PE}
Since $w$ is bounded \Eqref{eq:PETMPC.1.4}, the upper bound in Theorem \ref*{theo:PS.1} is trivially fulfilled \cite{Shouche1998}, therefore the focus is placed on achieving the lower bound. At time $i$ define,
\begin{subequations} \label{eq:PETMPC.4}
\begin{alignat}{1}
\boldsymbol{M}\left(w(i)\right)=\boldsymbol{M}_i&=\sum_{j=0}^{lp-1}\left(\boldsymbol{w}_{i-j}\boldsymbol{w}^{\top}_{i-j}\right)-\rho_{0}I_{mN_p} \label{eq:PETMPC.4.1}\\
\boldsymbol{w}_{i-j}&=\left[\begin{array}{c}
w(i-j)\\
w(i-j-1)\\
\vdots\\
w(i-j-N_p+1)
\end{array} \right] \label{eq:PETMPC.4.2}
\end{alignat}
\end{subequations}
That $\boldsymbol{M}_i$ depends only on the past and current exciting input makes the following a suitable \emph{persistence of excitation} constraint within the receding horizon fashion of MPC
\begin{equation} \label{eq:PETMPC.5}
\boldsymbol{M}_i>0
\end{equation}

\begin{rem}[] \label{rem:PETMPC.1}
The realization of constraint \eqref{eq:PETMPC.5}, i.e. a PE behaviour of $w$, does not necessarily imply persistence of excitation in the absolute input $u$. This is due to the disturbance rejection control policy \eqref{eq:PETMPC.3.3}. Numerical simulations have shown that a proper selection of the linear gain $K_t$ secures transmission of the PE condition towards $u$. The investigation of conditions for gain selection is beyond the scope of this paper, hence we require the following assumption.
\end{rem}

\begin{ass}[Persistence of excitation transmission] \label{ass:PETMPC.PE.1}
For the linear gain $K_t$, the persistence of excitation of $w$ is transmitted to the absolute input sequence $u$.
\end{ass}

\subsection{Tube MPC with additional PE constraint} \label{sec:PETMPC.R}
The optimization problem of the proposed model predictive controller, at time $i$ is, $\mathbb{P}_N\big(z=z(i)\big)$:
\begin{equation} \label{eq:PETMPC.6}
\min_{v,w}\smashoperator[r]{\sum_{k=0}^{N-1}}\left(z^{\top}(k)Qz(k)+d^{\top}(k)\mathcal{R}d(k)\right) + V_{f}\big(z(N)\big)
\end{equation}
subject to:
\begin{subequations} \label{eq:PETMPC.7}
\begin{alignat}{2}
z(0)&=z,\:\:z(k+1)=\bar{A}z(k&&)+\bar{B}v(k) \label{eq:PETMPC.7.2}\\
z(k) &\in \mathbb{Z},\:\:v(k)\in \mathbb{V},\:\: && k=0,\ldots,N-1 \label{eq:PETMPC.7.4}\\
z(N) &\in \mathbb{Z}_{f}\subseteq \mathbb{Z} \label{eq:PETMPC.7.5}\\
\tilde{w}(0) &\in \mathbb{W} \label{eq:PETMPC.7.6.1}\\
\tilde{w}(k) &=0_m,\: && k=1,\ldots,N-1 \label{eq:PETMPC.7.6.2}\\
\boldsymbol{M}&\left(w(i)=\tilde{w}(0)\right)>0 \label{eq:PETMPC.7.7}
\end{alignat}
\end{subequations}
where $\mathbb{Z}_f$ is the nominal equivalent $\mathbb{X}_f$. The cost term $d^{\top}(k)\mathcal{R}d(k)$ is a straight forward augmentation of the usual input cost with,
\begin{subequations} \label{eq:PETMPC.8}
\begin{align}
d^{\top}(k)&=[v^{\top}(k)\;\tilde{w}^{\top}(k)],\:\:k=0,1,\ldots,N-1 \label{eq:PETMPC.7.0.1}\\
\mathcal{R}&=\left[\begin{array}{cc}
R & \boldsymbol{0}_{m\times m}\\
\boldsymbol{0}_{m\times m} & R
\end{array} \right] \label{eq:PETMPC.7.0.2}\\
\bar{A}&=\tilde{A}(i),\quad\bar{B}=\tilde{B}(i) \label{eq:PETMPC.7.0.3}
\end{align}
\end{subequations}

\section{Stability and feasibility}  \label{sec:SF}
In this section stability and recursive feasibility proofs are derived for the PE-Tube MPC.

\subsection{Stability} \label{sec:SF.S}
The objective of performing closed-loop system identification is to reduce model uncertainty and thereby improve the performance of the MPC. Therefore it is safe to assume that the initially known model \Eqref{eq:PS.0.1} will differ from the actual plant being controlled, and from any transitional model \Eqref{eq:PS.0.3} given by the recursive identification \Eqref{eq:SF.5}. This represents a considerable drawback, as one of the main requirements of tube MPC is the computation of an RPI set, which is model dependent. Different approaches can be used to compute such sets; in \cite{Ferramosca2013a} for example, Assumption \ref*{ass:PETMPC.0} is bypassed by showing that, under a type of parametric affine model uncertainty, an RPI set computed for a certain model is also RPI for a family of models. The definition of a specific set of rules for computing $\mathbb{S}$ is out of the scope of this paper; for the examples shown in Section \ref*{sec:SR.1} a suitable RPI set is computed on the basis of Assumption \ref*{ass:PETMPC.0} and the error dynamics \Eqref{eq:PETMPC.3}.

\begin{theo}[Stability of the PE-Tube MPC] \label{theo:SF.1}
If assumptions \ref*{ass:PS.1}--\ref*{ass:PETMPC.2} hold, then the set $\mathcal{A}\coloneqq\mathbb{S}\times\{\boldsymbol{0}_n\}$ is asymptotically stable  with a region of attraction $\big(\mathcal{Z}_N\oplus\mathbb{S}\big)\times\mathcal{Z}_N$ for the constrained composite system,
\begin{equation*} 
\begin{aligned}
x(i+1)&=\bar{A}x(i)+\bar{B}\hat{u}(i)+\hat{w}(i)+w_S(i) \\
z(i+1)&=\bar{A}z(i)+\bar{B}v(i) \label{eq:SF.6.1}
\end{aligned}
\end{equation*}
under the closed-loop control laws defined by \Eqref{eq:PETMPC.3.3} and \Eqref{eq:PETMPC.6}--\Eqref{eq:PETMPC.7} respectively ($\mathcal{Z}_N$ is the nominal equivalent of $\mathcal{X}_N$).
\end{theo}

\begin{proof}
Follows directly from the stability proofs in \cite{Rawlings2014}.
\end{proof}

\subsection{Recursive feasibility} \label{sec:SF.F}
Recursive feasibility of a standard tube MPC formulation is provided by initializing inside $\mathcal{Z}_{N}\oplus\mathbb{S}$. In \cite{Marafioti2013} an additional assumption is used to provide a proof of recursive feasibility under the effects of the non-convex PE constraint \Eqref{eq:PETMPC.7.7},

\begin{ass}[] \label{ass:SF.2}
A feasible solution is available at time $i-1$, i.e., $\boldsymbol{M}_{i-1}>0$
\end{ass}

\begin{theo}[Recursive feasibility: trivial solution] \label{theo:SF.2}
If Assumption \ref*{ass:SF.2} holds, then there exists a feasible solution at time $i$ for the persistently exciting tube MPC \Eqref{eq:PETMPC.6}--\Eqref{eq:PETMPC.7}.
\end{theo}

\begin{proof}
The proof given in \cite{Marafioti2013} is repeated here for clarifying purposes. From \Eqref{eq:PETMPC.4},
\begin{subequations} 
\begin{alignat*}{1}
\boldsymbol{M}_i&=\boldsymbol{M}_{i-1}+\boldsymbol{w}_i\boldsymbol{w}^{\top}_i-\boldsymbol{w}_{i-l_p}\boldsymbol{w}^{\top}_{i-l_p}\\
\boldsymbol{w}^i&=\boldsymbol{w}_i\boldsymbol{w}^{\top}_i-\boldsymbol{w}_{i-l_p}\boldsymbol{w}^{\top}_{i-l_p}
\end{alignat*}
\end{subequations}
Clearly then, $\boldsymbol{w}^i \ge 0\implies\boldsymbol{M}_i > 0$. The proof is completed by noticing that $w(i)=w(i-l_p) \implies \boldsymbol{w}^i=0$.
\end{proof}

Theorem \ref*{theo:SF.2} provides recursive feasibility under the trivial periodic repetition of a previous solution, but it does not analyze the effect of choosing a different one. In fact, $w(i)$ is a decision variable in the proposed optimization problem \Eqref{eq:PETMPC.6}, hence the optimizer is \emph{free} to choose $w(i)\neq w(i-l_p)$ as long as constraints \Eqref{eq:PETMPC.7.6.1} and \Eqref{eq:PETMPC.7.7} are not violated. Numerical simulations (conducted on the same system used as an example in \cite{Marafioti2013}) show that recursive feasibility may be lost if periodicity is broken, namely, if the optimization algorithm lands in a solution such that $w(i)\neq w(i-l_p)$ . The observed behaviour can be summarized in,
\begin{equation*} 
\begin{aligned}
&\exists w(i)\in\mathbb{W}\;\;\mbox{s.t.},\;\;w(i)\neq w(i-l_p)\\
&\wedge\;\;\boldsymbol{M}_i>0,\;\;\mbox{but}\;\;\boldsymbol{M}_{i+1}\leq 0\;\;\forall w(i+1)\in\mathbb{W}
\end{aligned}
\end{equation*}
According to \Eqref{eq:PETMPC.4}, the non-trivial optimized $w(i)$ remains in $\boldsymbol{w}_i$ for $N_p-1$ time steps. To take this into account the following constraints are proposed to replace \Eqref{eq:PETMPC.7.6.2}--\Eqref{eq:PETMPC.7.7},
\begin{subequations} \label{eq:SF.4}
\begin{align}
\tilde{w}(k)=w(i+k-l_p),\:\:k&=1,\ldots,N_p-1 \label{eq:SF.4.2}\\
\tilde{w}(k)=0_m,\:\:k&=N_p,\ldots,N-1 \label{eq:SF.4.3}\\
\boldsymbol{M}\left(w(i+k)=\tilde{w}(k)\right)>0,\:\:k&=0,1,\ldots,N_p-1 \label{eq:SF.4.1}
\end{align}
\end{subequations}

\begin{theo}[Recursive feasibility: non-trivial solution] \label{theo:SF.3}
If Assumption \ref*{ass:SF.2} holds, and the constraint \Eqref{eq:PETMPC.7.7} is replaced by \Eqref{eq:SF.4}, then there exists a feasible, not necessarily trivial, solution at time $i$ for the PE-tube MPC \Eqref{eq:PETMPC.6}--\Eqref{eq:PETMPC.7}.
\end{theo}

\begin{proof}
This result is established by extending the proof for Theorem \ref*{theo:SF.2} to cover $N_p-1$ time steps.
\end{proof}

\begin{rem}[] \label{rem:PETMPC.2}
Each time instant the optimization updates the last element of the initially supplied PE sequence. Therefore, periodicity arises naturally in the optimized exciting sequence.
\end{rem}

\section{Simulation results} \label{sec:SR.1}
This section shows the behaviour of the proposed algorithm through two numerical examples. The task is to perform closed-loop system identification while regulating the states of the following multi-variable system (taken from \cite{Ferramosca2013a}):
\begingroup\setlength{\arraycolsep}{1pt}
\begin{subequations} \label{eq:SR.1}
\begin{alignat}{1}
A(\delta)=\bar{A}+\delta \hat{A}=&\begin{bmatrix*}[r]
0.42 &\:-0.28\\
0.02 &\:0.60
\end{bmatrix*}
+\delta
\begin{bmatrix*}[r]
-0.6 &\:-0.4\\
-0.6 &\:-0.85
\end{bmatrix*} \label{eq:SR.1.1}\\
B(\delta)=\bar{B}+\delta \hat{B}=&\begin{bmatrix*}[r]
0.30\\
-0.40
\end{bmatrix*}
+\delta
\begin{bmatrix*}[r]
-0.2\\
0.4
\end{bmatrix*} \label{eq:SR.1.2}
\end{alignat}
\end{subequations}
\endgroup
subject to the following constraints,
\begin{subequations} \label{eq:SR.2}
\begin{align}
\mathbb{X}&=\{x\in\mathbb{R}^2\:\text{s.t},\:|x_j|\leq 17,\:\:j=1,2\} \label{eq:SR.2.1}\\
\mathbb{U}&=\{u\in\mathbb{R}\:\text{s.t},\:|u|\leq 4\} \label{eq:SR.2.2}\\
\mathbb{W}&=\{w\in\mathbb{R}\:\text{s.t},\:|w|\leq 0.2\} \label{eq:SR.2.3}\\
|\delta|&\leq 0.15 \label{eq:SR.2.4}
\end{align}
\end{subequations}

Since no particular performance requirements are being considered, the controller parameters are loosely set to $N=3$, $Q=I_{2\times2}$ and $R=1$. The terminal cost $V_{f}$, and the terminal constraint set $\mathbb{Z}_f$ are computed according to Assumption \ref*{ass:PS.2}. $A(\delta)$ is inherently stable for any $\delta$ following \Eqref{eq:SR.2.4}, this provides flexibility in choosing the linear gain $K_t=[-0.112\;\;0.354]$ which is stabilizing and complies with Assumption \ref*{ass:PETMPC.PE.1}. A set $\mathbb{W}_S$ following Assumption \ref*{ass:PETMPC.2} is defined for \Eqref{eq:SR.1}--\Eqref{eq:SR.2} and the corresponding RPI set $\mathbb{S}$ is computed.

Following the directions given in \cite{Goodwin1984}, the PE constraint parameters are set to $N_p=6$ and $l_p=11$; given the size of $\mathbb{W}$ a value of $\rho_0=0.05$ is employed. To guarantee recursive feasibility, a feasible PE sequence of length $N_p+l_p-1$ has to be provided (Assumption \ref*{ass:SF.2}). However, this sequence only acts as a buffer for feasibility purposes, and the actual value feed to the system is an optimization variable.

For the RLS algorithm a forgetting factor $\lambda=0.97$ is employed. The estimates vector $\tilde{\theta}_{j}$ is initialized at the known values $(\bar{A},\bar{B})$. The information matrix $R_{\text{ID}}$ is initialized as the null matrix, therefore a \emph{pseudo}-inverse is computed for \Eqref{eq:PS.5.1} until $R_{\text{ID}}$ becomes invertible. Albeit a recursion is computed at every time instant, the prediction model update \Eqref{eq:PETMPC.7.0.3} is performed only every 3 time steps.

\subsection{Closed-loop identification capabilities} \label{sec:SR.1.1}
An initial state $x(0)=\left[0\;0\right]^{\top}\in\mathcal{Z}_N$ is considered to assess the closed-loop identification capabilities of the proposed algorithm. This is done to avoid the additional information that would be generated in the process of regulation, so that the effect of the PE constraint can be observed independently. Fig. \ref*{fig:SR.3} shows the optimized input signal generated by the PE-Tube MPC for both, nominal system ($v$) and true plant ($u$). As expected, given the initial state, the nominal input remains at the origin while the input for the uncertain system is, indeed, disturbed by the PE part $w$. Fig. \ref*{fig:SR.2} shows the trajectory of the states $x_1$ and $z_1$. During the initial time steps, the algorithm optimizes a PE sequence on the basis of the feasible sequence initially supplied, hence the transient behaviour observed on the uncertain system ($x$). Feasibility is maintained during this period thanks to the tighter constraint \Eqref{eq:SF.4}. An optimized periodic solution is attained fairly fast.

\begin{figure}[!t]
\centering
\pgfplotsset{
every axis legend/.append style={
at={(0.95,0.94)},
anchor=north east,
}}
\setlength\figureheight{0.15\textwidth}
\setlength\figurewidth{0.4\textwidth}
\begin{tikzpicture}

\begin{axis}[%
width=0.95092\figurewidth,
height=\figureheight,
at={(0\figurewidth,0\figureheight)},
scale only axis,
xmin=0,
xmax=100,
xlabel={Time instant $i$},
ymin=-0.3,
ymax=0.5,
ylabel={Input signal: $u$ and $v$},
legend style={legend plot pos=left},
xtick={0,20,...,100},
ytick={-0.3,-0.2,-0.1,0,0.1,0.2,0.3,0.4,0.5},
minor x tick num={1},
minor y tick num={1},
]
\addplot [color=blue,solid]
  table[row sep=crcr]{%
0	0\\
1	0.0290368727132807\\
2	0.165297896447218\\
3	-0.234104158551057\\
4	0.0212040549054089\\
5	0.0472044941043568\\
6	-0.0204229655733163\\
7	-0.0471883710551028\\
8	0.0313216301256673\\
9	0.17760381506334\\
10	-0.000464898161530412\\
11	-0.0249512218679178\\
12	0.0173098733886392\\
13	0.143057971160174\\
14	-0.23180577717915\\
15	0.0231297760280962\\
16	0.0484195963182072\\
17	-0.0196386838846428\\
18	0.00334676763553338\\
19	0.0214990242032553\\
20	0.17180868166137\\
21	-0.00259580693833744\\
22	-0.0256814412045859\\
23	0.0171878968576737\\
24	0.114517844885253\\
25	-0.225951615135102\\
26	0.0261631631578641\\
27	0.0497322536446976\\
28	-0.0191213376562695\\
29	0.00346088292133868\\
30	0.0214606350006435\\
31	0.171734249338206\\
32	-0.00268236731104337\\
33	-0.025722792601029\\
34	0.0171626688510444\\
35	0.114504223643203\\
36	-0.22595811370938\\
37	0.026160548019601\\
38	0.0497272006023915\\
39	-0.0191204544679188\\
40	0.00346157095511306\\
41	0.0214610812420785\\
42	0.171734503814466\\
43	-0.0026822384547688\\
44	-0.0257227360665299\\
45	0.0171626884441284\\
46	0.114504226835029\\
47	-0.225958116341744\\
48	0.0261605443157399\\
49	0.0497228874645483\\
50	-0.0191195957354257\\
51	0.00346200950486922\\
52	0.0214612759961676\\
53	0.171734572957614\\
54	-0.00268222575154672\\
55	-0.0257227438513955\\
56	0.0171626765092632\\
57	0.114504216805223\\
58	-0.225958123136245\\
59	0.0261605403011924\\
60	0.0497185757923553\\
61	-0.0191187360387941\\
62	0.00346244882570819\\
63	0.0214614712581563\\
64	0.171734642394067\\
65	-0.00268221289792728\\
66	-0.0257227515691076\\
67	0.0171626645984907\\
68	0.114504206780062\\
69	-0.225958129933231\\
70	0.0261605362826605\\
71	0.0497142646428518\\
72	-0.0191178764495222\\
73	0.00346288809151366\\
74	0.0214616664955848\\
75	0.171734711821719\\
76	-0.0026822000459935\\
77	-0.0257227592858993\\
78	0.0171626526891859\\
79	0.114504196756145\\
80	-0.225958136729367\\
81	0.0261605322646334\\
82	0.0497099540196213\\
83	-0.0191170169651757\\
84	0.00346332730370185\\
85	0.0214618617091836\\
86	0.171734781240899\\
87	-0.00268218719562612\\
88	-0.0257227670017467\\
89	0.0171626407813373\\
90	0.114504186733455\\
91	-0.225958143524671\\
92	0.0261605282470993\\
93	0.0497056439225932\\
94	-0.0191161575857423\\
95	0.0034637664622796\\
96	0.021462056898957\\
97	0.171734850651609\\
98	-0.00268217434682313\\
99	-0.0257227747166478\\
100	0.0171626288749465\\
101	0.114504176711993\\
};
\addlegendentry{Disturbed system $u$};

\addplot [color=red,dashed]
  table[row sep=crcr]{%
0	0\\
1	1.8102582531006e-08\\
2	1.86701072246059e-08\\
3	1.80734787005638e-08\\
4	1.80734787005638e-08\\
5	-1.75964346556109e-09\\
6	-5.86575990119485e-08\\
7	-4.7062222888131e-08\\
8	-3.76930812945619e-08\\
9	-3.87790401689215e-08\\
10	-3.87790401689215e-08\\
11	-3.87790401689215e-08\\
12	-3.87790401689215e-08\\
13	-2.937156295242e-08\\
14	-3.15188653772183e-08\\
15	-3.25457694670463e-08\\
16	-3.23070053935093e-08\\
17	-3.23448916292741e-08\\
18	9.98578713431183e-08\\
19	8.5567890588862e-08\\
20	6.39812339833261e-08\\
21	6.40060390392652e-08\\
22	6.40963831328903e-08\\
23	6.41115273899031e-08\\
24	3.66326952851708e-08\\
25	4.2276210359837e-08\\
26	4.57350730665329e-08\\
27	4.55081530330818e-08\\
28	4.55548614124913e-08\\
29	4.55779817545066e-08\\
30	4.55779817545066e-08\\
31	4.55779817545066e-08\\
32	4.54966140544184e-08\\
33	4.55092729711379e-08\\
34	4.55211675261706e-08\\
35	4.55211675261706e-08\\
36	4.55211675261706e-08\\
37	4.55211675261706e-08\\
38	4.54991113914742e-08\\
39	4.55015263485974e-08\\
40	4.55050293242847e-08\\
41	4.55050293242847e-08\\
42	4.55050293242847e-08\\
43	4.55050293242847e-08\\
44	4.55050293242847e-08\\
45	4.55050293242847e-08\\
46	4.55050293242847e-08\\
47	4.55050293242847e-08\\
48	4.55050293242847e-08\\
49	4.54829758541235e-08\\
50	4.54853925876035e-08\\
51	4.54888946751124e-08\\
52	4.54888946751124e-08\\
53	4.54888946751124e-08\\
54	4.54888946751124e-08\\
55	4.54888946751124e-08\\
56	4.54888946751124e-08\\
57	4.54888946751124e-08\\
58	4.54888946751124e-08\\
59	4.54888946751124e-08\\
60	4.54668483103786e-08\\
61	4.54692632675018e-08\\
62	4.54727653550106e-08\\
63	4.54727653550106e-08\\
64	4.54727653550106e-08\\
65	4.54727653550106e-08\\
66	4.54727653550106e-08\\
67	4.54727653550106e-08\\
68	4.54727653550106e-08\\
69	4.54727653550106e-08\\
70	4.54727653550106e-08\\
71	4.54507260957042e-08\\
72	4.54531410528274e-08\\
73	4.54566413639794e-08\\
74	4.54566413639794e-08\\
75	4.54566413639794e-08\\
76	4.54566413639794e-08\\
77	4.54566413639794e-08\\
78	4.54566413639794e-08\\
79	4.54566413639794e-08\\
80	4.54566413639794e-08\\
81	4.54566413639794e-08\\
82	4.54346083219219e-08\\
83	4.54370223908667e-08\\
84	4.54405218138403e-08\\
85	4.54405218138403e-08\\
86	4.54405218138403e-08\\
87	4.54405218138403e-08\\
88	4.54405218138403e-08\\
89	4.54405218138403e-08\\
90	4.54405218138403e-08\\
91	4.54405218138403e-08\\
92	4.54405218138403e-08\\
93	4.54184967653885e-08\\
94	4.54209099461549e-08\\
95	4.54244084809501e-08\\
96	4.54244084809501e-08\\
97	4.54244084809501e-08\\
98	4.54244084809501e-08\\
99	4.54244084809501e-08\\
100	4.54244084809501e-08\\
101	4.54244084809501e-08\\
};
\addlegendentry{Nominal system $v$};

\end{axis}
\end{tikzpicture}%
\caption{Optimized input signal for the nominal system and true plant, $x(0)=\left[0\;0\right]^{\top}$.} \label{fig:SR.3}
\end{figure}
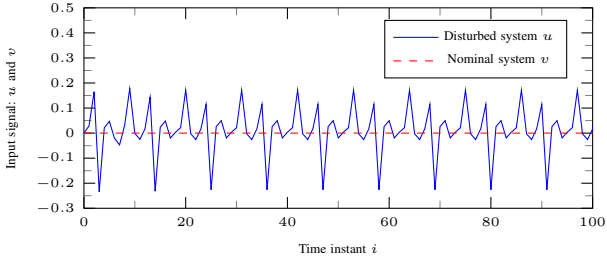

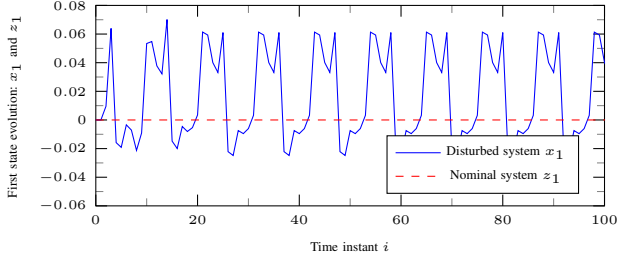
\begin{figure}[!t]
\centering
\pgfplotsset{
every axis legend/.append style={
at={(0.95,0.06)},
anchor=south east,
}}
\setlength\figureheight{0.15\textwidth}
\setlength\figurewidth{0.4\textwidth}
\begin{tikzpicture}

\begin{axis}[%
width=0.95092\figurewidth,
height=\figureheight,
at={(0\figurewidth,0\figureheight)},
scale only axis,
xmin=0,
xmax=100,
xlabel={Time instant $i$},
ymin=-0.06,
ymax=0.08,
ylabel={First state evolution: $x_{\scriptscriptstyle 1}$ and $z_{\scriptscriptstyle 1}$},
legend style={legend plot pos=left},
xtick={0,20,...,100},
ytick={-.06,-.04,-.02,0,.02,.04,.06,.08},
minor x tick num={1},
minor y tick num={1},
scaled ticks = false,
y tick label style={
/pgf/number format/.cd,
fixed,
precision=2,
/tikz/.cd
},
]
\addplot [color=blue,solid]
  table[row sep=crcr]{%
0	0\\
1	0\\
2	0.00958216799538264\\
3	0.0639765783975841\\
4	-0.0158282550035736\\
5	-0.019131096013411\\
6	-0.00340684250653387\\
7	-0.00709176439593472\\
8	-0.0212478895930748\\
9	-0.00911318460851451\\
10	0.0533890249241519\\
11	0.0547765551486405\\
12	0.037785571629049\\
13	0.0321875820624136\\
14	0.0701602875367055\\
15	-0.0147891395831154\\
16	-0.0199276011765471\\
17	-0.00457706934588571\\
18	-0.00805746820886043\\
19	-0.00535402716709791\\
20	0.00347990003909116\\
21	0.061391028942142\\
22	0.0593174124147685\\
23	0.0400739954156336\\
24	0.0331857557171709\\
25	0.0610528426649364\\
26	-0.022147313146284\\
27	-0.0248024620454282\\
28	-0.0074250893839305\\
29	-0.00952999030199897\\
30	-0.00601870157466179\\
31	0.00323656326221931\\
32	0.0613400188755568\\
33	0.0593313354821109\\
34	0.0401056559762749\\
35	0.0332143839680948\\
36	0.0610729551045532\\
37	-0.0221351212677249\\
38	-0.0247959082462861\\
39	-0.00742340720775937\\
40	-0.00952987107893115\\
41	-0.00601908792855776\\
42	0.00323612614403903\\
43	0.0613396823763122\\
44	0.0593311188235581\\
45	0.0401055331703558\\
46	0.0332143222197328\\
47	0.0610729282618138\\
48	-0.0221351304017962\\
49	-0.0247959095884645\\
50	-0.00742482816368984\\
51	-0.00953098468018411\\
52	-0.00601981278920423\\
53	0.00323571152477933\\
54	0.0613394718573784\\
55	0.059331026053759\\
56	0.0401055007169772\\
57	0.0332143166610751\\
58	0.0610729323265601\\
59	-0.0221351245043194\\
60	-0.0247959047067518\\
61	-0.00742624704715372\\
62	-0.00953209806322499\\
63	-0.00602053805247737\\
64	0.00323529642123236\\
65	0.061339260957874\\
66	0.0593309330358877\\
67	0.0401054681216199\\
68	0.0332143110303692\\
69	0.0610729363596508\\
70	-0.0221351186177984\\
71	-0.0247958998267748\\
72	-0.00742766575547009\\
73	-0.00953321130797666\\
74	-0.00602126322532997\\
75	0.00323488136960009\\
76	0.0613390500848348\\
77	0.059330840029744\\
78	0.0401054355304112\\
79	0.0332143054004147\\
80	0.0610729403922653\\
81	-0.0221351127319976\\
82	-0.0247958949473992\\
83	-0.00742908429060857\\
84	-0.00953432441683638\\
85	-0.00602198830966053\\
86	0.00323446636863484\\
87	0.0613388392375395\\
88	0.0593307470349566\\
89	0.0401054029431843\\
90	0.0332142997711509\\
91	0.0610729444243911\\
92	-0.0221351068469118\\
93	-0.0247958900686156\\
94	-0.00743050265258218\\
95	-0.00953543738981406\\
96	-0.00602271330547476\\
97	0.00323405141833435\\
98	0.0613386284159881\\
99	0.0593306540515271\\
100	0.0401053703599414\\
101	0.0332142941425804\\
};
\addlegendentry{Disturbed system $x_{\scriptscriptstyle 1}$};

\addplot [color=red,dashed]
  table[row sep=crcr]{%
0	0\\
1	0\\
2	5.43077475930179e-09\\
3	1.04695500264994e-08\\
4	1.39522500002765e-08\\
5	1.62479021585629e-08\\
6	1.36605526481437e-08\\
7	-8.94088504399364e-09\\
8	-2.72660154130152e-08\\
9	-3.85907984265532e-08\\
10	-4.62631053862213e-08\\
11	-5.10413931376959e-08\\
12	-5.38209622601466e-08\\
13	-5.53091068742796e-08\\
14	-5.29109728339869e-08\\
15	-5.08312669513737e-08\\
16	-4.93962916362911e-08\\
17	-4.83051952149989e-08\\
18	-4.75735242283684e-08\\
19	-3.48363052661124e-09\\
20	3.50010203659777e-08\\
21	5.86904784702027e-08\\
22	7.41025384933043e-08\\
23	8.35273338335151e-08\\
24	8.88932855148714e-08\\
25	8.26150468934563e-08\\
26	7.68157890350007e-08\\
27	7.28726006134503e-08\\
28	6.99971628845995e-08\\
29	6.80837067252226e-08\\
30	6.68939267716174e-08\\
31	6.62011428163064e-08\\
32	6.58297585491603e-08\\
33	6.56262694644858e-08\\
34	6.55371202263325e-08\\
35	6.55127740418304e-08\\
36	6.55180611392095e-08\\
37	6.55345475126212e-08\\
38	6.55527900169711e-08\\
39	6.55613531445936e-08\\
40	6.55671418031088e-08\\
41	6.55717617830616e-08\\
42	6.55748216639085e-08\\
43	6.55767213511443e-08\\
44	6.55778224822637e-08\\
45	6.55784092522648e-08\\
46	6.55786856506635e-08\\
47	6.55787880421444e-08\\
48	6.55788018635859e-08\\
49	6.55787771483792e-08\\
50	6.55714632731354e-08\\
51	6.55650652929938e-08\\
52	6.55610884307689e-08\\
53	6.55584871216326e-08\\
54	6.55568922515306e-08\\
55	6.55559810394635e-08\\
56	6.55555047993518e-08\\
57	6.5555287614002e-08\\
58	6.55552133554994e-08\\
59	6.55552101645288e-08\\
60	6.55552365447316e-08\\
61	6.55479951931426e-08\\
62	6.55416649116366e-08\\
63	6.55377423546364e-08\\
64	6.55351805984004e-08\\
65	6.55336124954011e-08\\
66	6.55327182532302e-08\\
67	6.55322520810112e-08\\
68	6.55320404237489e-08\\
69	6.55319688942752e-08\\
70	6.5531966822669e-08\\
71	6.55319934734815e-08\\
72	6.55247543443867e-08\\
73	6.55184261081468e-08\\
74	6.55145045882435e-08\\
75	6.55119434287137e-08\\
76	6.55103756405621e-08\\
77	6.55094815443054e-08\\
78	6.55090154240636e-08\\
79	6.55088037715281e-08\\
80	6.5508732226675e-08\\
81	6.55087301341062e-08\\
82	6.5508756765259e-08\\
83	6.55015196720889e-08\\
84	6.54951931500083e-08\\
85	6.54912727010739e-08\\
86	6.5488712241081e-08\\
87	6.54871448811959e-08\\
88	6.54862510292106e-08\\
89	6.54857850363399e-08\\
90	6.54855734416595e-08\\
91	6.54855019163807e-08\\
92	6.54854998244056e-08\\
93	6.54855264482983e-08\\
94	6.54782919837915e-08\\
95	6.54719677556829e-08\\
96	6.54680488531466e-08\\
97	6.54654894442242e-08\\
98	6.54639227534511e-08\\
99	6.54630293002085e-08\\
100	6.54625635274998e-08\\
101	6.54623520424133e-08\\
};
\addlegendentry{Nominal system $z_{\scriptscriptstyle 1}$};

\end{axis}
\end{tikzpicture}%
\caption{First state trajectory for the nominal system and true plant, $x(0)=\left[0\;0\right]^{\top}$.} \label{fig:SR.2}
\end{figure}

Table \ref*{table:SR.1} shows the error of the identified value (w.r.t. true plant) for all the system parameters, at several time instants. Thanks to the PE effect of $w$, and in view of the deterministic framework, the RLS algorithm has enough information to achieve perfect parameter estimation.

\setlength{\tabcolsep}{3pt}
\begin{table}[!t] 
\caption{Evolution of estimated parameters error [\%].} \label{table:SR.1}
\centering
\begin{tabular}{>{$}r<{$} | >{$}r<{$} >{$}r<{$} >{$}r<{$} >{$}r<{$}}
\hline
\text{Parameter} & i=0 & i=3 & i=6 & i=9 \\
\hline
A_{11} & -17.6 & -1.76\times 10^{1}   & -2.83\times 10^{-11} & -2.24\times 10^{-11} \\
A_{12} & -17.6 & -1.76\times 10^{1}   &  1.25\times 10^{-11} &  9.48\times 10^{-12} \\
A_{21} & -81.8 & -8.18\times 10^{1}   & -5.09\times 10^{-10} & -4.04\times 10^{-10} \\
A_{22} & -17.5 & -1.75\times 10^{1}   & -2.27\times 10^{-11} & -1.71\times 10^{-11} \\
B_{11} & -9.09 &  1.68\times 10^{-14} & -2.35\times 10^{-13} &  4.37\times 10^{-13} \\
B_{21} & -13.0 &                  0   &  6.87\times 10^{-13} & -1.19\times 10^{-12} \\
\hline
\end{tabular}
\end{table}

\begin{figure}[!t]
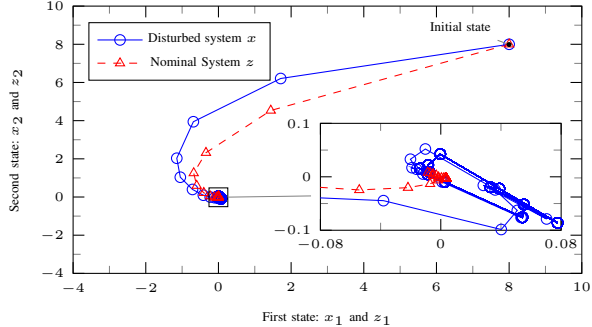

\centering
\pgfplotsset{
every axis legend/.append style={
at={(0.03,0.94)},
anchor=north west,
}}
\setlength\figureheight{0.2\textwidth}
\setlength\figurewidth{0.4\textwidth}
%
%
}}] at (pt) {};

\end{tikzpicture}%
\caption{State trajectory for the nominal system and true plant.} \label{fig:SR.6}
\end{figure}

\subsection{Regulation capabilities} \label{sec:SR.1.2}
The regulation capabilities of the PE-Tube MPC are evaluated by initializing the scheme at $x(0)=\left[8\;8\right]^{\top}\in\mathcal{Z}_N$; Fig. \ref*{fig:SR.6} shows the state trajectory for both, nominal system and true plant. As expected, given the stabilizing characteristics of the proposed algorithm, the nominal state shows an asymptotic behaviour towards the origin. Due to the periodic PE disturbance ($w$), the state of the true plant is ultimately bounded to lie inside the set $\{\boldsymbol{0}_n\}\oplus\mathbb{S}$.

\section{Conclusions and future work} \label{sec:CFW.1}
In this paper, a new way to approach the dual problem of system identification and regulation via a model predictive controller has been presented. At each time instant, the input used to control the system is divided into a persistently exciting (PE) part and a regulator part. The PE part is treated as a bounded disturbance and a tube MPC, enhanced with a PE constraint, is used to regulate the plant. At the same time, thanks to the PE constraint, enough information is generated for convergence of the estimated parameters. Under the proper assumptions, the PE-Tube MPC has proved robust stability and recursive feasibility. Future work will be focused on the analysis of the transmission of persistence of excitation, the implementation of on-line variation of the allowable size of perturbations with the objective of feasibility enlargement and the shifting of the proposed approach to an stochastic framework.


\bibliographystyle{IEEEtran}
\bibliography{IEEEabrv,references_V2}

\end{document}